%% file: paper.tex
\documentclass[conference]{IEEEtran}

\usepackage{listings}
\usepackage{epsfig}
\usepackage{amsmath}
\usepackage{amssymb}
\usepackage{multicol}
\usepackage{numprint}
\usepackage{siunitx}
\usepackage{algorithm}
\usepackage[noend]{algpseudocode}
\usepackage{color}

\usepackage{amsthm}

\usepackage[lofdepth,lotdepth]{subfig}
\usepackage{flushend} 

\npthousandsep{,}\npthousandthpartsep{}\npdecimalsign{.}

\newtheorem{theorem}{Theorem}
\newtheorem{lemma}{Lemma}

\newtheorem{condition}{Condition}

\ifx\DRAFT\undefined
\newcommand{\todo}[1]{\TODOsLETFINNONDRAFTDOCUMENT}
\else
\newcommand{\todo}[1]{{\bf \textcolor{red}{#1}}}
\fi

\newcommand{\etal}{et al.}

\input{defines}  

\begin{document}

\title{Annotating Control-Flow Graphs for Formalized Test Coverage Criteria}

\ifx\undefined\ANONYMIZED

\author{\IEEEauthorblockN{Sean Kauffman}
\IEEEauthorblockA{Electrical and Computer Engineering\\
Queen's University\\
Kingston, Ontario, Canada\\
Email: sean.k@queensu.ca}
\and
\IEEEauthorblockN{Carlos Moreno}
\IEEEauthorblockA{Electrical and Computer Engineering\\
University of Waterloo\\
Waterloo, Ontario, Canada\\
Email: cmoreno@uwaterloo.ca}
\and
\IEEEauthorblockN{Sebastian Fischmeister}
\IEEEauthorblockA{Electrical and Computer Engineering\\
University of Waterloo\\
Waterloo, Ontario, Canada\\
Email: sfischme@uwaterloo.ca}}

\else

\author{\\\;\\Anonymized Author(s)\\\;}

\fi

\maketitle

\input{abstract}
\input{introduction}
\input{related}
\input{definitions}
\input{problem}
\input{modeling}
\input{criteria}
\input{tool}

\input{conclusion}

\bibliographystyle{IEEEtran}
\bibliography{paper}

\end{document}

%% file: defines.tex
\def\xor{\oplus}
\def\scand{\text{\&\&}}
\def\scor{||}

\def\Section#1{Section~\ref{#1}}

\newcommand{\powerset}[1]{2^{#1}}
\newcommand*{\Cdot}{\raisebox{-0.25ex}{\scalebox{1.6}{.}}}
\newcommand{\where}[1][1]{\Cdot\ } 

\DeclareFontFamily{U}{mathx}{}
\DeclareFontShape{U}{mathx}{m}{n}{
      <5> <6> <7> <8> <9> <10>
      <10.95> <12> <14.4> <17.28> <20.74> <24.88>
      mathx10
      }{}
\DeclareSymbolFont{mathx}{U}{mathx}{m}{n}
\DeclareFontSubstitution{U}{mathx}{m}{n}

\DeclareMathSymbol{\bigtimes}{1}{mathx}{"91}

%% file: abstract.tex
\begin{abstract}
Control flow coverage criteria are an important part of the process of qualifying embedded software for safety-critical systems.  Criteria such as modified condition/decision coverage (MC/DC) as defined by DO-178B are used by regulators to judge the adequacy of testing and by QA engineers to design tests when full path coverage is impossible.

Despite their importance, these coverage criteria are often misunderstood.  One problem is that their definitions are typically written in natural language specification documents, making them imprecise.  Other works have proposed formal definitions using binary predicate logic, but these definitions are difficult to apply to the analysis of real programs.  Control-Flow Graphs (CFGs) are the most common model for analyzing program logic in compilers, and seem to be a good fit for defining and analyzing coverage criteria.  However, CFGs discard the explicit concept of a decision, making their use for this task seem impossible.

In this paper, we show how to annotate a CFG with decision information inferred from the graph itself.  We call this annotated model a Control-Flow Decision Graph (CFDG) and we use it to formally define several common coverage criteria.  We have implemented our algorithms in a tool which we show can be applied to automatically annotate CFGs output from popular compilers.

\end{abstract}

%% file: introduction.tex
\section{Introduction}
\label{sec:introduction}

Control flow coverage criteria are an important part in the process of qualifying embedded software for safety-critical systems.  Standards such as DO-178B/C in aerospace systems and ISO-26262 in the automotive domain prescribe modified condition/decision coverage (MC/DC) as a prerequisite for software test design.
Coverage criteria are used during software qualification to guide test design and demonstrate that no part of the code base is overlooked.

Despite their importance, coverage criteria are often misunderstood.  
One problem is that their definitions are typically written in natural language specification documents, making them imprecise.  
Other works have proposed formal definitions using binary predicate logic or \(Z\) notation, but these definitions can be difficult to apply to the analysis of real programs.
One of the most common representations of program control flow is the Control-Flow Graph (CFG), which is used by most compilers for optimization and analysis, and is easy to represent visually.
As such, modeling coverage criteria using CFGs seems like a natural fit to improve user comprehension and simplify automated measurement.

Unfortunately, CFGs do not explicitly retain a concept that is crucial to measuring control-flow coverage criteria: \emph{decisions}.
Decisions represent potentially complex branch conditions in programs, and they may be represented as multiple vertices in a CFG.
To define coverage criteria using a CFG model and represent it for easier human and machine comprehension, decisions need to be made explicit. 

In this paper, we propose an intermediate representation of program logic that extends CFGs to represent decisions explicitly called Control-Flow Decision Graphs (CFDGs).  We show how to automatically transform a CFG into a CFDG and we use this structure to formalize common control flow test coverage criteria.
Our method is implemented as a tool that can be used to annotate CFGs output by the popular GCC and Clang compilers in \textit{dot} format, allowing them to be visually compared.

The remainder of the paper is structured as follows: 
\Section{sec:related} discusses related work.  
\Section{sec:definitions} defines terms and notation used in the paper.
\Section{sec:motivation} presents the context of the problem and why it is important.
\Section{sec:model} defines a Control-Flow Decision Graph.
\Section{sec:criteria} uses the model to formally define test coverage criteria.
\Section{sec:tool} introduces a tool to annotate CFGs output by popular compilers.
\Section{sec:conclusion} concludes the paper.

%% file: related.tex
\section{Related Work}
\label{sec:related}

Previous works have formalized coverage criteria requirements for comparison purposes.  Comar \etal{} formalized MC/DC coverage using Branch Decision Diagrams to compare with Object Branch Coverage~\cite{couverture}.  Vilkomir and Bowen used \(Z\) notation to formalize and compare traditional test coverage criteria including MC/DC in~\cite{vilkomir2001formalization,vilkomir2001formalization2}, and added RC/DC in~\cite{formalization}.  They applied the same techniques to formalize other regulatory requirements in the \(Z\) notation in~\cite{vilkomir2001application}.  Woodward and Hennell compared JJ-Path, or LCSAJ, coverage to MC/DC and defined when JJ-Path subsumes MC/DC~\cite{jjpaths}.  Kirner proposed a formal model for MC/DC to be used to show that transformation and compilation do not result in coverage violations~\cite{preserving}.  Ammann \etal{} formalized a comprehensive set of criteria using boolean predicate logic~\cite{ammann2003coverage}.  Kapoor and Bowen used boolean logic to compare the fault detection effectiveness of MC/DC and RC/DC~\cite{kapoor2005formal}.  Kosmatov \etal{} formalized data-oriented boundary testing heuristics for a new family of model-based coverage criteria~\cite{kosmatov2004boundary}.

Other works have formalized coverage criteria for automatic test case generation.  In 1994, Jasper \etal{} proposed using their Ada framework (TSDT) to generate test cases that met MC/DC requirements~\cite{jasper1994test}.  They modeled the criterion using two functions on tuples of boolean expressions and vectors.  Ghani and Clark proposed using search-based test generation to meet MC/DC requirements~\cite{ghani2009automatic}.  They used an AST-based emitter to export decisions in disjunctive normal form (DNF), then used them in cost functions for simulated annealing.  Offutt, Xiong and Liu introduced a method of test code generation coverage criteria such as FPC from SCR condition tables or UML state charts~\cite{offutt1999criteria}.



%% file: definitions.tex
\section{Definitions and Notation}
\label{sec:definitions}

In this section, we present definitions of terms and notation used throughout the paper.

\subsection{Conditions and Decisions}

A \textit{condition} is an expression defined on a Boolean (\(\mathbb{B}\)) algebra with elements \emph{true} and \emph{false}.
A \textit{decision} is a function from two or more conditions and binary operators to a two-valued boolean algebra \(\mathbb{B}\).  The operators in a decision may be any of the following: normal AND (\(\wedge\)), short-circuit AND (\&\&), normal OR (\(\vee\)), short-circuit OR (\,\(||\)\,), and XOR (\(\xor\)).  For example, the decision \((a \xor b)\) combines two conditions (reading the values of symbols $a$ and $b$) with an XOR.

A decision is equivalent to an \verb+if+ statement in imperative programming, or any equivalent structures that make decisions about the control flow of the program such as a \verb+while+ statement.

\subsection{Test Coverage}

An $\textit{input}$ is a mapping from a symbol in a program to a concrete, true or false, value.  One input can map to one or more conditions.  For example, the decision \((a \xor b \vee a)\) combines three conditions, two of which $a$ are mapped to the same input and one which $b$ is mapped to a second input.

A \emph{test} $t \in \mathbb{B}^*$ is a finite sequence of Boolean values mapping the inputs of a program to concrete values.  Note that we assume a fixed order for program symbols, permitting the omission of explicit symbol mapping.  A \emph{test suite} $T \in \powerset{\mathbb{B}^*}$ is a set of tests.

\subsection{Control Flow Graph}

A \textit{control flow graph} (CFG) of a program is a connected, directed graph \(G = (V, E)\) where each vertex \(v \in V\) has outdegree \(\leq 2\) and represents a program point and each edge \((t,h) \in E\) represents the possibility that execution after the end of the tail (\(t\)) may continue with the beginning of the head (\(h\))~\cite{controlflow}.  The domain of vertices is given by \(\mathcal{V}\), so the domain of CFGs is \(\mathcal{G} = \mathcal{V} \times (\mathcal{V} \times \mathcal{V})\).


Each edge represents possible program flow and each vertex (a program point) is equivalent to a single statement.  A CFG has a single \textit{entry vertex}, given by the function \({\textit{entry} : \mathcal{G} \rightarrow \mathcal{V}}\), which is a special vertex with indegree 0 from which the program flow begins. A CFG has a set of one-or-more \textit{exit vertices}, given by the function \(\textit{exit} : \mathcal{G} \rightarrow \mathcal{V}\), which are special vertices with outdegree 0 at which the program flow ends.  A vertex in a CFG \textit{dominates} another vertex if all paths from the entry vertex to the dominated vertex must pass through the dominating vertex.  The entry vertex dominates every other vertex in the CFG.
A vertex has successor vertices given by the function $\textit{successors} : \mathcal{V} \rightarrow \powerset{\mathcal{V}}$ defined as $\textit{successors}(v) = \{ s \in V : (v,s) \in E \}$.


\subsection{Program Runs}


We define $\textit{Run} : \mathcal{G} \times \mathbb{B}^* \rightarrow E^* \cup E^\omega$ to be a function from a CFG and a test to the sequence of edges in the CFG \((t,h) \in E\) that are traversed. 
The sequence, called a (lowercase) \emph{run}, can be finite or infinite.
Runs returned by \textit{Run} must begin with the entry node and, if finite, end with an exit node.

A run may walk each edge zero or more times.  Walking an edge zero times means that the condition represented by the vertex before the edge never resulted in a choice of that edge when executing the program with the given input.  Edges may be walked more than once due to loops in the program.  A run may be a finite sequence in the case of a terminating program or an infinite sequence in the case of a non-terminating program.

Listing \ref{lst:scand} shows a program with a single decision containing a short-circuit AND operator, and two conditions, each with a separate input.  Figure~\ref{fig:runs} shows three copies of the CFG for this program.  Each of the graphs contains a run of the program, with the executed edges marked as dashed lines and highlighted in red.  Below each graph, the state of both inputs is given.

\begin{lstlisting}[language=C,frame=single,caption=A program with one decision and two conditions,label=lst:scand]
  x = 0;

  if (a && b) {
    x = 1;
  }

  return x;
\end{lstlisting}

For example, in Figure \ref{fig:runs1}, \(a\) is true, so the vertex with \(b\) is reached.  Since \(b\) is false, the vertex with \(x = 1\) is not reached and the program instead continues to the final statement.

\begin{figure}[ht]
\centering
  \subfloat[][$\neg a, b \vee \neg a, \neg b$]{
    \includegraphics[width=.125\textwidth]{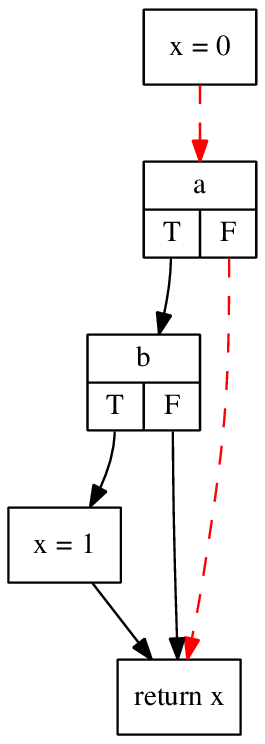}
    \label{fig:runs0}
  }
  \subfloat[][$a, \neg b$]{
    \includegraphics[width=.125\textwidth]{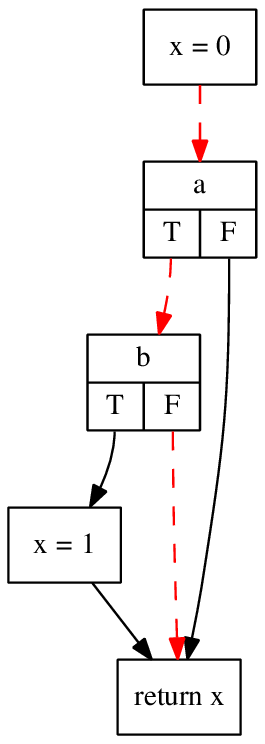}
    \label{fig:runs1}
  }
  \subfloat[][$a, b$]{
    \includegraphics[width=.125\textwidth]{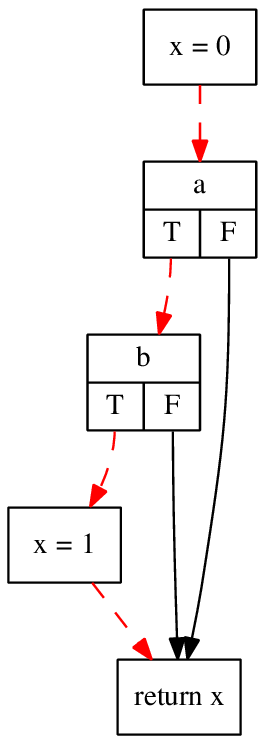}
    \label{fig:runs3}
  }
\caption{Possible runs of the program in Listing \ref{lst:scand}}
\label{fig:runs}
\end{figure}

%% file: problem.tex
\section{Motivation}
\label{sec:motivation}

The most widely used and accepted model for program control flow is the CFG.  CFGs are relatively simple but contain all of the necessary information to analyze program control flow properties.  However, CFGs do not explicitly define the concept of a decision, making them difficult to use to model coverage criteria.

Control flow test coverage criteria specify the control flow properties of the tests in a coverage.  They are typically defined in natural language documents which are open to interpretation by engineers and regulators and are often misunderstood and misapplied.  A natural fit for describing and analyzing these coverage criteria is a CFG, but many such criteria rely on the definition of decisions, which CFGs lack.  For example, a commonly used criteria is \emph{decision} coverage, which specifies that executing a test suite must result in the program taking all possible outcomes from every decision.

Prior methods of formalizing control flow test coverages have typically used boolean predicate logic to construct solvers to either test that a coverage meets a criterion or to automatically generate test coverages that do so.  
These efforts are useful in that they have formalized the criteria so that meeting them is unambiguous.
However, these efforts do not facilitate finding other properties of test coverage criteria, or give users other methods to reason about them.
We believe that modifying CFGs to carry decision coverage information provides more natural way to express formal coverage criteria and can help users to reason about those criteria in their test development.

%

One idea that has been proposed is to use object code coverage in place of other criteria when information about decisions is lacking.
CFGs can even be inferred from dynamic analysis of compiled binaries~\cite{rimsa2021cfggrind}.
While it is true that object code (i.e., the assembly-level output of a compiler) is closer to what the system will execute, object code discards information about the intent of the program that may be important and that 
is present in the source code. In particular, object code coverage 
does not change the MC/DC requirements~\cite{couverture}.  The Federal Aviation Administration (FAA) allows the use of object code coverage only when it is shown to be equivalent to source code coverage~\cite{cast}. 
This means that the coverage criteria must use source code, not object code, to define the notion of a decision.

%% file: modeling.tex
\section{Control-Flow Decision Graphs}
\label{sec:model}

\def\Ds{D_\mathrm{s}}

In this section, we construct a model that can be used to formally describe the control flow of a program and many control-flow test coverage criteria.
We show how to infer decision information from a CFG to construct such a model. 
We later use this model to formalize the test coverage criteria in Section~\ref{sec:criteria}.
Note that, throughout this section, we assume that conditions are represented by single vertices and that decisions contain no extra, interstitial vertices.
We define an interstitial vertex, in this case, as a non-condition vertex that appears between condition vertices within a decision.
In Section~\ref{sec:tool}, we show that these assumptions do not always hold when using CFGs output by certain tools.


To model control flow through a decision in the DFG we can use the common practice from interprocedural analysis~\cite{interprocedural} of adding a subgraph to a CFG which usually represents a called function.  The subgraph of the decision only has one entry edge, leading to one entry vertex which dominates the subgraph.
That there is only one entry vertex is easy to see from the definition of a CFG and a decision.
Incoming edges to a CFG vertex represent the possibility of control flow transitioning to the beginning of that vertex.  

We define a \emph{Control-Flow Decision Graph} (CFDG) as a CFG with an additional set of decision subgraphs that group vertices of the CFG.
We create a CFDG from a CFG by subgraphs of conditions in the \emph{shape} of decisions. 
The CFG is then represented by the CFDG = \(( G, \Ds ) \), where \(G = (V, E)\) is the original CFG, and $\Ds$ is the set of decision subgraphs. 
Once these subgraphs are created we can reason about them independently, and understand their relations.


The idea of inferring decisions from CFGs comes from an observation of their shape.
First, consider that any vertex in a CFG with outdegree 2 is a condition.
All conditions are part of a decision (that may contain only one condition).
Of the five Boolean operators we consider for decisions ($\wedge$, $\scand$, $\vee$, $\scor$, and $\xor$), $\wedge$, $\vee$ and $\xor$ only require one CFG vertex.
This makes sense if one considers the effect that no short-circuiting has on evaluation - both sides of the operator must always be evaluated so there is no alternative control flow that would split the vertex (basic block).
As such, we need only to consider multiple vertices for the short-circuit operators: $\scand$ and $\scor$.

\begin{lemma}[Outdegree of Decision Subgraphs]
A decision can be represented as a control flow subgraph that has two successor vertices.
\label{lem:outdegree}
\end{lemma}

\begin{proof}
The statement is proved by construction, observing that a decision involves a conditional expression that evaluates to a binary result, and control continues at one of two possible vertices, depending on this binary result.
Thus, we can always construct a subgraph for a decision with exactly two successor vertices.
\end{proof}

\begin{lemma}[Indegree of Decision Subgraphs]
A decision can be represented as a control flow subgraph which is dominated by a single entry vertex.
\label{lem:indegree}
\end{lemma}

\begin{proof}
We will prove the statement by induction on the number of conditions. The base case is a decision with one condition. The statement is trivially true in this case since the graph contains a single vertex representing 
the condition.

As induction hypothesis, we assume that the statement is true for decisions consisting of $n$ conditions.

We will show the induction step by contradiction. Consider a decision with $n+1$ conditions and assume, for the purpose of a contradiction, that it is not dominated by a single vertex.

The lexical rules we have established require us to express the decision as a condition combined via a binary operator with a decision with $n$ conditions. 
If the graph for a decision with $n+1$ conditions is not dominated by a single vertex, then there must be an incoming edge to a vertex in the decision subgraph. 
The resulting structure cannot happen with natural conditional 
constructs, and would require a label right after the first condition (to be targeted by a \verb+goto+ statement), which is not allowed.
\end{proof}

\begin{theorem}[Shared Successors of Conditions in a Decision]
Assuming no interstitial vertices \emph{within a decision} in a CFG, the subgraph that corresponds to a decision has one predecessor and two successors, and all conditions in that decision must share a successor with the subgraph.
\label{thm:successors}
\end{theorem}

\begin{proof}
The first part of the statement follows directly from lemmas~\ref{lem:outdegree}~and~\ref{lem:indegree}.

The second part can be shown by induction on the number of conditions. The base case is a decision with one condition. The statement is trivially true in this case since the decision subgraph contains a single vertex representing the condition, so they must share all successors.

As induction hypothesis, we assume that the statement is true for decisions consisting of $n$ conditions.

We will show the induction step by contradiction.
Assume a decision with $n+1$ conditions such that the first condition shares no successors with the decision subgraph.
Since the decision can only have two successors and the final condition's successors are those two, the first condition's successors must be vertices within the decision subgraph.
Since, by definition, a condition has two successors, the two vertices within the decision must be different.
This contradicts Lemma~\ref{lem:indegree} that says the decision with $n$ conditions must be dominated by a single vertex, as it would add an incoming edge to a vertex in the decision that does not pass through the entry vertex.
\end{proof}
%

\subsection{Algorithm}
\label{sec:algorithm}

Figure~\ref{fig:runs} shows a CFG with one decision including an $\scand$ operator.
In the CFG, there are two decisions, $a$ and $b$, that must be true to reach the vertex with $x = 1$.
If either is \emph{false}, the outgoing edge leads to the \emph{same} vertex containing $\text{return } x$.
This same shape is true for any $\scand$ decision; any condition that is \emph{false} leads to the same successor vertex.
Crucially, $\scor$ decisions are shaped identically, except that the equal successor vertex is reached on a \emph{true} condition.
This leads to the algorithms below.
Note that we overload the $\textit{successor} : \powerset{\mathcal{V}} \rightarrow \powerset{\mathcal{V}}$ function for sets of vertices that represent decision subgraphs, defined as $\textit{successor}(D) = \{ s \in V : (c,s) \in E \wedge c \in D \}$ for a CFDG $(G,V),\Ds$.

The procedure for annotating a CFG to create a CFDG is given in Algorithms~\ref{alg:init}~and~\ref{alg:alg}.
The initial procedure is given in Algorithm~\ref{alg:init}, which takes a CFG as its input and returns the CFDG.
The algorithm first sets up a data structure to ensure that vertices are only visited once, which is necessary in case of loops.
It then sets up a data structure $D_{\text{map}}$ that maps condition vertices to a set of vertices that represents a decision.
Then, the algorithm iterates over these sets, calling \textit{merge} on each one, so long as the vertex that maps to it has not yet been visited.
Note that the parameters, especially \textit{visited} and $D_{\text{map}}$, should be considered passed-by-reference.

Algorithm~\ref{alg:alg} works by recursing over subsequent decisions, merging each one reverse order and returning the successor vertices of the decision.
On Line~\ref{alg:succloop}, the successors of the decision $D_1$ are iterated over.  This should be understood as the union of all the successors to the vertices in $D_1$ (a set).
If a successor has not been visited, it is checked to see if it is a condition.
If it is, then its current decision is looked up in the map and is recursed on, getting the set of its successors from the return.
Then, on Line~\ref{alg:test}, if the successors returned by the recursive call contain a shared successor with $D_1$ (not the current successor we iterated over), then the decisions are merged, and the map is updated for all contained vertices.

\begin{algorithm}[ht]
\caption{CFDG Creation Algorithm}
\label{alg:init}
\begin{algorithmic}[1]
\Procedure{createCFDG( $(V,E)$ )}{} \label{alg:create}
\State $\textit{visited} \gets \{ v \mapsto \textit{false} : v \in V \}$
\State{$D_{\text{map}} \gets \{ v \mapsto \{v\} : v \in V \where \textit{outdegree}(v) = 2 \}$}
\For{$v \mapsto D_1 \in V$}
  \If{\textbf{not} \textit{visited}($v$)}
    \State{\textit{merge}( $D_1, D_{\text{map}}, \textit{visited}$ )}
  \EndIf
\EndFor \\
\Return{$(V,E), \{ \Ds : v \mapsto \Ds \in D_{\text{map}} \}$}
\EndProcedure
\end{algorithmic}
\end{algorithm}

\begin{algorithm}[ht]
\caption{Decision Merging Algorithm}
\label{alg:alg}
\begin{algorithmic}[1]
\Procedure{merge( $D_1, D_{\text{map}}, \textit{visited}$ )}{} \label{alg:merge}
\For{$\textit{s} \in \textit{successors}(D_1)$} \Comment{$s$ is a vertex} \label{alg:succloop}
  \If{\textbf{not} \textit{visited}($s$)} 
    \State{\textit{visited}($s$) $\gets \textit{true}$}
    \If{\textit{outdegree}($s$) = 2} \Comment{$s$ is a condition}
      \State{$D_2 \gets D_{\text{map}}(s)$}
      \State{$S \gets $ \textit{merge}( $D_2, D_{\text{map}}, \textit{visited}$ )} \label{alg:recurse}
      \If{$\textit{successors}(D_1) \setminus s \cap S \neq \varnothing$} \label{alg:test}
        \State{$D_1 \gets D_1 \cup D_2$} \Comment{merge!}
        \For{$v \mapsto \_ \in D_1$}
          \State{$D_{\text{map}}(v) \gets D_1$}
        \EndFor
      \EndIf
    \EndIf
  \EndIf
\EndFor \\
\Return $\textit{successors}(D_1)$
\EndProcedure
\end{algorithmic}
\end{algorithm}

\subsection{Example of constructing a CFDG}

\begin{lstlisting}[language=C,frame=single,caption=One decision with three conditions,label=lst:three]
  x = ((a && b) || c);
\end{lstlisting}

Consider the program in Listing~\ref{lst:three}, which contains one decision with a short-circuit AND as well as a short-circuit OR, and three conditions.  Figure \ref{fig:cfdg1} shows the CFG for this program where the vertices representing the conditions are shown with rounded corners.

\begin{figure}[ht]
\centering
  \subfloat[][Step 1]{
    \includegraphics[width=.13\textwidth]{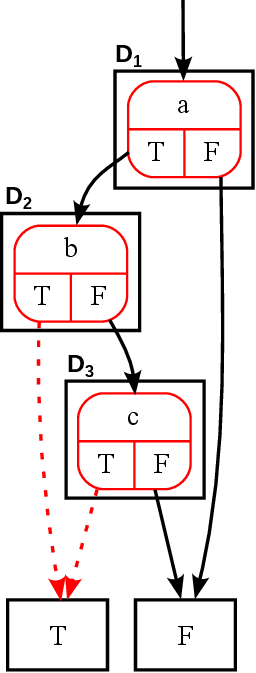}
    \label{fig:cfdg1}
  }
  \subfloat[][Step 2]{
    \includegraphics[width=.13\textwidth]{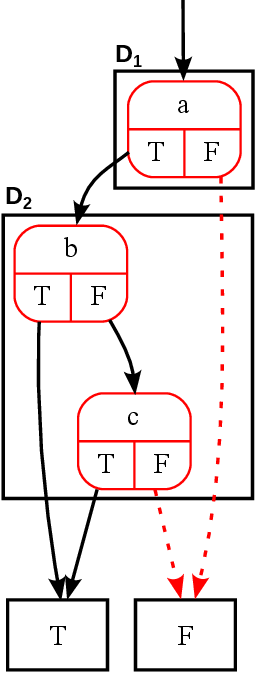}
    \label{fig:cfdg2}
  }
  \subfloat[][Step 3]{
    \includegraphics[width=.13\textwidth]{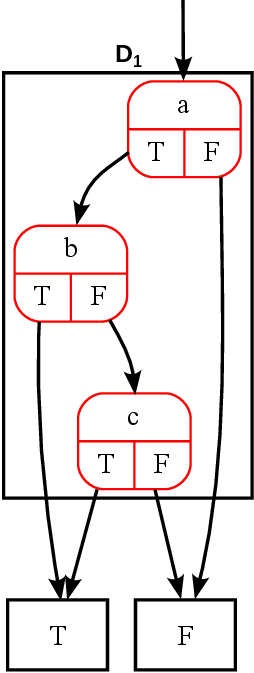}
    \label{fig:cfdg3}
  }
\caption{Steps of Algorithms~\ref{alg:init}~and~\ref{alg:alg} on the CFG for Listing~\ref{lst:three}}
\label{fig:cfdg}
\end{figure}

At the beginning of the algorithm, Algorithm~\ref{alg:init} iterates over the conditions, forming the decision subgraphs $D_i$ for $i \in \{1,2,3\}$, which are represented by black, solid-line boxes around the condition vertices.
The algorithm then calls \textit{merge} on Decision~$D_1$.

Algorithm~\ref{alg:alg} then recurses forward in the CFG, until it reaches the last decision $D_3$ containing the condition vertex, $c$.
The recursive call to \textit{merge} on $D_3$ (Line~\ref{alg:recurse}) returns the set of its successor vertices, $\{T,F\}$.
This set is then compared to the successors of $D_2$, subtracting the entry vertex of $D_3$, $c$.
The intersection of those two successor sets is non-empty, since the $T$ vertex is shared by both sets.
This is shown in Figure~\ref{fig:cfdg1}, where the two edges pointing at the shared successor are highlighted in red and dashed.
As a result of this shared successor, $D_2$ and $D_3$ are merged, and the recursive call to \textit{merge} returns the successors of that merged decision, $\{T,F\}$.
Like before, and as shown in Figure~\ref{fig:cfdg2}, the successors of $D_1$ minus $b$ are compared to $\{T,F\}$ with the intersection found non-empty.
As such, $D_1$ and $D_2$ are merged, with the final decision shown in Figure~\ref{fig:cfdg3}.

\subsection{Correctness and Complexity}

Algorithm~\ref{alg:init} on its own is trivially correct, \emph{assuming that all conditions are represented as single vertices in a CFG}, but we must argue for the correctness of Algorithm~\ref{alg:alg}.

\begin{theorem}[Soundness]
Algorithm~\ref{alg:alg} is sound.
\label{thm:sound}
\end{theorem}

\begin{proof}
The proof of soundness for Algorithm is based on Theorem~\ref{thm:successors}, which says that all conditions in a decision must share a successor with the decision itself.
Algorithm~\ref{alg:alg} merges two decisions $D_1$ and $D_2$ when $D_1$'s successors are the entry vertex of $D_2$ (there is only one, by Lemma~\ref{lem:indegree}) and one of the successors of $D_2$.
This forms a new decision which is clearly sound, as it maintains the same successors, and still has only one entry vertex (that of $D_1$).
\end{proof}

\begin{theorem}[Completeness]
Assuming no interstitial vertices \emph{within a decision}, Algorithm~\ref{alg:alg} is complete.
\label{thm:complete}
\end{theorem}

\begin{proof}
By our assumption that every condition is represented as a single vertex, clearly Algorithm~\ref{alg:init} will create decisions for every condition.  
So, need only show that every pair of decisions will be merged if they are really in one decision.
It is clear that every decision is compared with its successor conditions to see if they belong in the same decision.
For a decision to not be merged, it must not share a successor with its predecessor decision, in which case it is not in the same decision by Theorem~\ref{thm:successors}.
\end{proof}

\begin{theorem}[Correctness]
Algorithm~\ref{alg:alg} is correct.
\label{thm:correct}
\end{theorem}

\begin{proof}
By Theorem~\ref{thm:sound} Algorithm~\ref{alg:alg} is sound, and by Theorem~\ref{thm:complete} it is complete.
\end{proof}

\begin{theorem}[Complexity]
Algorithms~\ref{alg:init}~and~\ref{alg:alg} run in linear time complexity in the size of the CFG.
\end{theorem}

\begin{proof}
Because of the \textit{visited} data structure, clearly vertices in a CFG are visited at most twice.  This can happen if a vertex is the last condition in a decision and is visited from a call to \textit{merge} from Algorithm~\ref{alg:init} before being visited in a recursive call in Algorithm~\ref{alg:alg}.

Edges are only visited when they are outgoing edges from a vertex being visited (in calls to \textit{successors}, with a maximum of two edges visited per call and a maximum of two calls.
\end{proof}

%% file: criteria.tex
\section{Test Coverage Criteria Definitions}
\label{sec:criteria}

\newcommand{\Program}{$\mathcal{P}$}
\newcommand{\TestCov}{\mathcal{T}}
\newcommand{\Runs}{R}

This section uses the CFDG model described in Section \ref{sec:model} to formally define common test coverage criteria.

Given a program \Program{}, its CFDG \((V,E), \Ds\), a function that defines the set of entry vertices of a graph \(\textit{entry}( V,E ) = \{v \in V : \nexists (\cdot,v) \in E\}\), and a function that defines a set of exit vertices of a graph \(\textit{exit}(V,E) = \{v \in V : \nexists (v,\cdot) \in E\}\), a test suite \(\TestCov\), and the set of test runs \(\Runs = \{\mathrm{Run}(G,t) : t \in \TestCov\}\), we define the following test coverage criteria.

\vskip 2mm
\subsection{Statement Coverage (SC)}

Statement coverage (SC) is the simplest criterion, and only requires that every statement in a program has been executed at least once~\cite{myers2011art}.

\(\TestCov\) meets 100\% SC requirements for \Program{} if it meets the following condition:

\begin{condition}[\(\TestCov\) contains test cases which visit all of the vertices]
\label{cond:sc}
\[\forall v \in V \where \exists \, \textbf{r} \in \Runs \, \where (\cdot,v) \in \textbf{r} \vee (v,\cdot) \in \textbf{r} \]
\end{condition}

\vskip 2mm
\subsection{Decision Coverage (DC)}

Decision coverage (DC) expands on SC by stipulating that every statement in a program has executed at least once, and every decision in the program has taken all possible outcomes at least once~\cite{myers2011art}.
Recall that, by Lemma~\ref{lem:outdegree}, decision subgraphs in CFDGs have only two successors.

\(\TestCov\) meets 100\% DC requirements for \Program{} if it meets Condition \ref{cond:sc}, and the following condition:

\begin{condition}[\(\TestCov\) contains test cases which visit both of the outgoing edges from every decision]
\label{cond:dc}
\[\forall D \in \Ds \where \exists \, \textbf{r} \in \Runs \where \exists v_1,s_1,v_2,s_2 \in V \where\] \[(v_1,s_1),(v_2,s_2) \in \textbf{r} \wedge s_1 \neq s_2 \wedge v_1,v_2 \in D \wedge s_1,s_2 \not\in D\]
\end{condition}

\vskip 2mm
\subsection{Condition Coverage (CC)}

Condition coverage (CC), also called branch coverage, is a stronger expansion on SC, and specifies that every statement in a program has been executed at least once, and every condition in every decision has taken all possible outcomes at least once~\cite{myers2011art}.

\(\TestCov\) meets 100\% CC requirements for \Program{} if it meets Condition \ref{cond:sc}, and the following condition:

\begin{condition}[\(\TestCov\) contains test cases which take both outgoing edges from every condition vertex in every decision subgraph]
\label{cond:cc}
\[\forall D \in \Ds \where \forall v \in D \where \exists \, \textbf{r} \in \Runs \where \exists s_1,s_2 \in V \where\] \[(v,s_1),(v,s_2) \in \textbf{r} \wedge s_1 \neq s_2\]
\end{condition}

\vskip 2mm
\subsection{Decision/Condition Coverage (D/CC)}

Decision/condition coverage (D/CC) combines the requirements of DC and CC.  D/CC specifies that every statement in a program has been executed at least once, every decision in the program has taken all possible outcomes at least once, and every condition in every decision has taken all possible outcomes at least once~\cite{myers2011art}.
Note that CC does not imply DC because changing condition outcomes may not imply changing decision outcomes.  To see this, suppose a decision $(c_1\; \text{\&\&}\; c_2)$.  If $c_1$ is varied while $c_2$ is held false, then the decision outcome does not change.

\(\TestCov\) meets 100\% D/CC coverage requirements for \Program{} if it meets conditions \ref{cond:sc}, \ref{cond:dc}, and \ref{cond:cc}.

\vskip 2mm
\subsection{Multiple Condition Coverage (MCC)}

Multiple condition coverage (MCC) is a still stronger criterion.  MCC requires that every statement in a program has been executed at least once, and all possible combinations of the conditions in every decision have been taken at least once~\cite{myers2011art}.

\(\TestCov\) meets 100\% MCC requirements for \Program{} if it meets conditions \ref{cond:sc}, and the following condition:

\begin{condition}[\(\TestCov\) contains test cases which, for every condition in every decision, vary only that condition]
\[\forall D \in \Ds \where \forall v \in D \where \exists \, \textbf{r}_1,\textbf{r}_2 \in \Runs \where \exists v,s_1,s_2,c,x \in V \where\]
\[(v,s_1) \in \textbf{r}_1 \wedge (v,s_2) \in \textbf{r}_2 \wedge s_1 \neq s_2\; \wedge\] \[\{ (c,x) \in \textbf{r}_1 : c \in D \wedge c \neq v \} = \{ (c,x) \in \textbf{r}_2 : c \in D \wedge c \neq v \}\]
\end{condition}

\vskip 2mm
\subsection{Full Predicate Coverage (FPC)}

Full predicate coverage (FPC) requires that every statement in a program has been executed at least once, and every condition in every decision has taken all possible outcomes where the outcome is directly correlated to the outcome of the decision.  This is a weaker criterion than MC/DC (introduced below) in that it does not require other conditions to be held fixed when varying a condition~\cite{offutt1999criteria, vilkomir2001formalization, vilkomir2001formalization2}.

\(\TestCov\) meets 100\% FPC requirements for \Program{} if it meets conditions \ref{cond:sc}, and the following condition:

\begin{condition}[\(\TestCov\) contains test cases which, for every condition in every decision, vary that condition and result in a change to the outcome of the decision]
\[\forall D \in \Ds \where \forall v \in D \where \exists \, \textbf{r}_1,\textbf{r}_2 \in \Runs \where \exists s_1,s_2,c,x_1,x_2 \in V \where\]
\[(v,s_1) \in \textbf{r}_1 \wedge (v,s_2) \in \textbf{r}_2 \wedge s_1 \neq s_2\; \wedge \] \[(c,x_1) \in \textbf{r}_1 \wedge (c,x_2) \in \textbf{r}_2 \wedge c \in D \wedge \{x_1,x2\} = \textit{successors}(D)\]
\end{condition}

\vskip 2mm
\subsection{Modified Condition/Decision Coverage (MC/DC)}

Modified condition/decision coverage (MC/DC) is a criterion for qualifying the completeness of tests used in development standards in the avionics and automotive industries, notably in DO-178B and ISO-26262.  It specifies that every entry and exit point to a program is tested, that every possible outcome of a decision is tested, that every possible result of every condition in a decision is tested, and that every condition in a decision independently affects the decision's outcome~\cite{mcdc, couverture, investigation}.

\(\TestCov\) meets 100\% MC/DC requirements for \Program{} if it meets conditions \ref{cond:dc}, \ref{cond:cc}, and the following conditions:

\begin{condition}[\(\TestCov\) contains a test case which visits the entry vertex]
\[\forall v \in \textit{entry}(V,E) \where \exists \, \textbf{r} \in \Runs \where (v,\cdot) \in \textbf{r}\]
\end{condition}

\begin{condition}[\(\TestCov\) contains test cases which visit all of the exit vertices]
\[\forall v \in \textit{exit}(V,E) \where \exists \, \textbf{r} \in \Runs \where (\cdot,v) \in \textbf{r}\]
\end{condition}

\begin{condition}[\(\TestCov\) contains test cases which, for every condition in every decision, vary \emph{only} that condition \emph{and} result in a change to the outcome of the decision]
\[\forall D \in \Ds \where \forall v \in D \where \exists \, \textbf{r}_1,\textbf{r}_2 \in \Runs \where \exists s_1,s_2,c,x,o,x_1,x_2 \in V \where\]
\[(v,s_1) \in \textbf{r}_1 \wedge (v,s_2) \in \textbf{r}_2 \wedge s_1 \neq s_2\; \wedge\] \[\{ (c,x) \in \textbf{r}_1 : c \in D \wedge c \neq v \} = \{ (c,x) \in \textbf{r}_2 : c \in D \wedge c \neq v \}\]
\[\wedge (o,x_1) \in \textbf{r}_1 \wedge (o,x_2) \in \textbf{r}_2 \wedge o \in D \wedge \{x_1,x_2\} = \textit{successors}(D)\]
\end{condition}
\vskip 5mm



%% file: tool.tex
\section{A Tool to Annotate CFGs}
\label{sec:tool}

\begin{figure}[b!]
\centering
    \includegraphics[width=.4\textwidth]{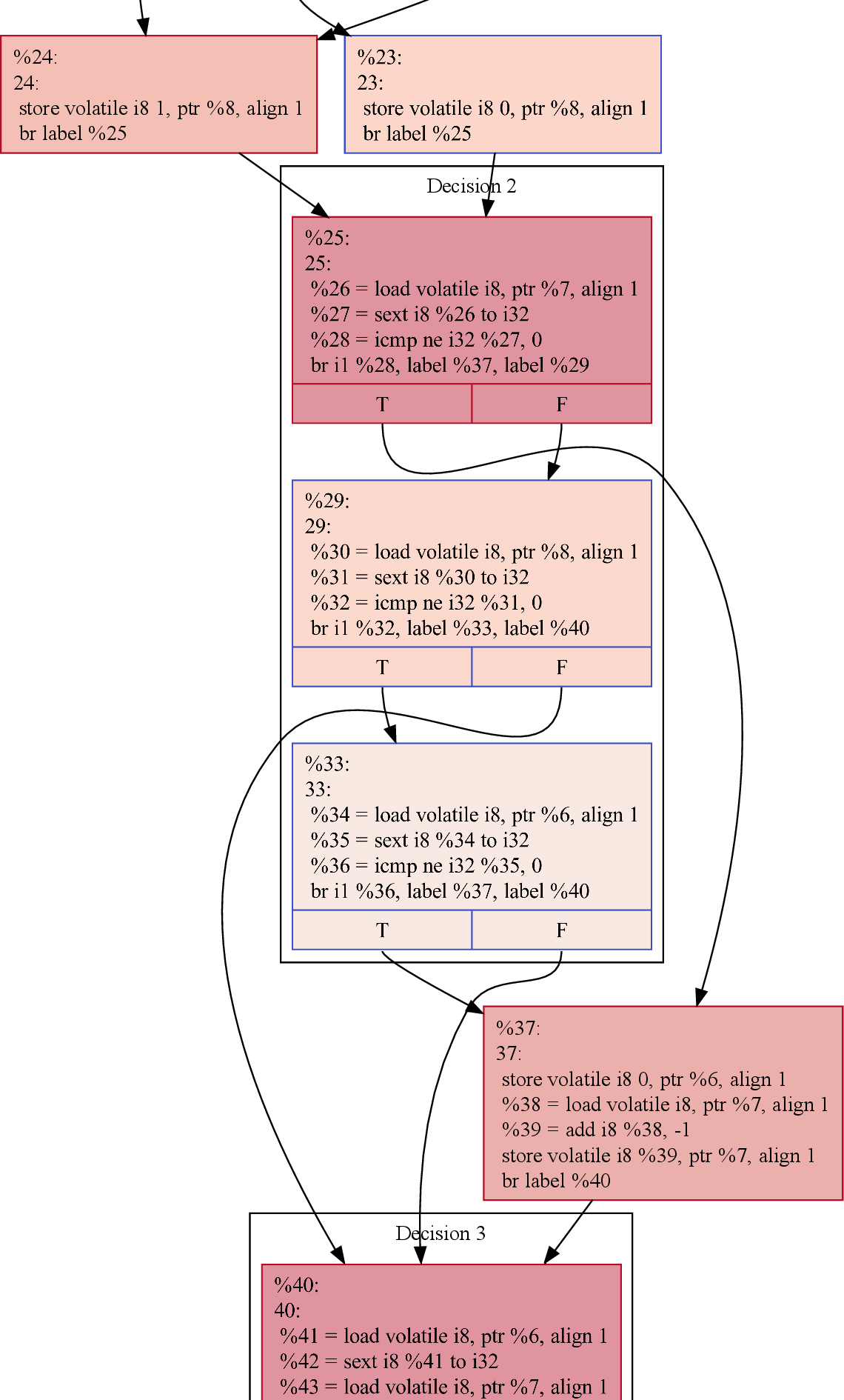}
\caption{Parts of a visualized CFG output from Clang annotated using cfg2cfdg}
\label{fig:clang}
\end{figure}

We implemented the algorithms from Section~\ref{sec:algorithm} in an open-source tool available at~\cite{tool} called \emph{cfg2cfdg}.
The tool is written in Python and uses the PyGraphViz~\cite{pygraphviz} library to read and write GraphViz files in \emph{dot} format.
GraphViz dot files are a common format for representing graphs for analysis and visualization, including for CFGs.
The popular compilers GNU Compiler Collection (GCC) and Clang are able to write GraphViz dot files containing CFGs for functions written in their supported languages and cfg2cfdg can annotate these files with decision information.

\begin{figure}[h!]
\centering
    \includegraphics[width=.4\textwidth]{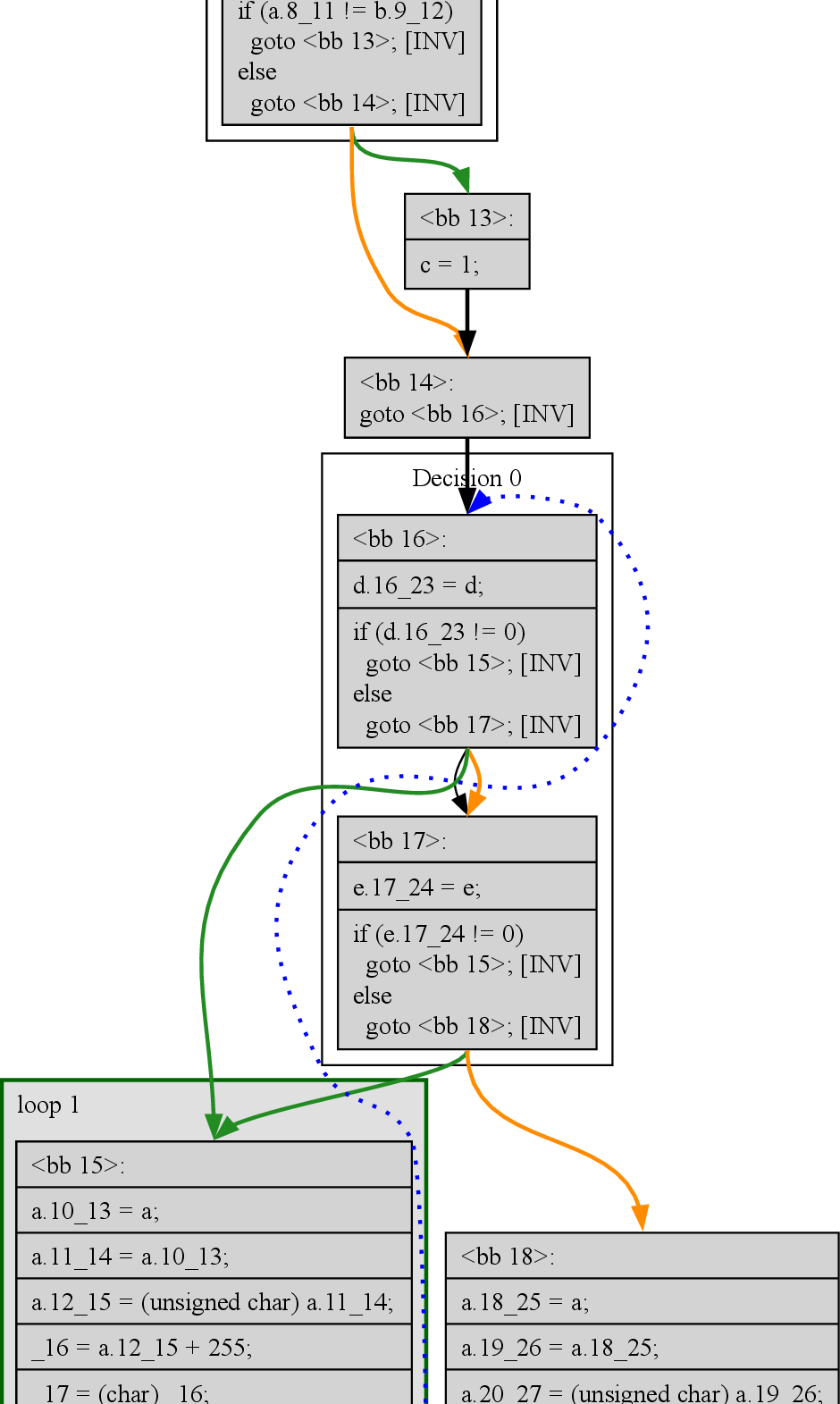}
\caption{Parts of a visualized CFG output from GCC annotated using cfg2cfdg}
\label{fig:gcc}
\end{figure}

Visualizations of annotated outputs from both compilers are shown in Figures~\ref{fig:clang}~and~\ref{fig:gcc}.
In the figure, vertices in the CFDG are shown as boxes, with arrows showing the edges.
The decision annotations are visible in both graphs as boxes around the vertices they encapsulate, with the labels \hbox{\textbf{Decision \#}}.

Figure~\ref{fig:clang} shows part of a CFDG from Clang with a complex decision (Decision~2) with three conditions combined with an $\scand$ and an $\scor$ operator similar to Listing~\ref{lst:three}.
Clang names each basic block in the CFG with a comment on the first line in the form \texttt{\%n} before the label \texttt{n:}.
In the figure, Decision~2 has two entry vertices, \texttt{\%24} and \texttt{\%23}, and two successor vertices, \texttt{\%37} and \texttt{\%40}.
Decision~2 can be understood as a conditional on $\texttt{\%25}\; \scor\; (\texttt{\%29}\; \scand\; \texttt{\%33})$, where basic-block \texttt{\%37} is reached if the decision is true, with execution continuing to \texttt{\%40} (no \textbf{else} block).

Figure~\ref{fig:gcc} shows part of a CFDG from GCC with a loop where Decision~0 has two conditions, where the back edge is a dotted blue line.
GCC names each basic block by its label: \texttt{<bb n>:}.
In the figure, Decision~0 has two entry vertices, \texttt{<bb 14>} and \texttt{<bb 15>} (also labeled \texttt{loop 1}), which has a back-edge creating the loop, and two successor vertices, \texttt{<bb 15>} and \texttt{<bb 18>}.
Decision~0 can be understood as a while loop on $\texttt{<bb 16>}\; \scor\; \texttt{<bb 17>}$.
When the loop terminates, it continues with \texttt{<bb 18>}.

\subsection{Challenges}

Working with the output from these tools means relying on their specific interpretations of how CFGs should behave, which does not always align with expectations from the theory.
One challenge we encountered was that Clang adds interstitial vertices to loops, separating \emph{guards} and \emph{latches}.
This means that decisions do not meet the assumptions stated in Section~\ref{sec:model} and, as such, the algorithm is not complete and decisions are not properly merged.
GCC, on the other hand, treats loops as we expected, and our tool has no trouble adding decisions around loop decisions as in Figure~\ref{fig:gcc}.

Using the tool is generally simple, requiring only the dot filename as input and immediately writing out a new dot file with decision annotations.
Visualization of the new CFDG using open-source GraphViz software is generally easy, but we did find that in some circumstances it was necessary to edit the files before visualization.
For example, GCC can use subgraphs in the file to annotate function boundaries, and this can interfere with visualizing decisions.
We found that editing the function subgraph names to remove the keyword \emph{cluster} solved the problem.

%% file: conclusion.tex
\section{Conclusion}
\label{sec:conclusion}

We have proposed a model, called a Control-Flow Decision Graph (CFDG), for formalizing and reasoning about control flow coverage criteria.  We introduced an algorithm to construct a CFDG from a CFG and proved its correctness and complexity, and then implemented that algorithm in an open-source tool for annotating CFGs from GCC and Clang~\cite{tool}.
We also formalize common test coverage criteria using the CFDG model, showing how explicit information about decisions in a CFG can help with program understanding and analysis.

It is our hope that this work will provide the basis for further insight into the relationship of decisions to program logic and control-flow properties.